\documentclass[runningheads]{llncs}

\usepackage{graphicx} 
\usepackage[dvipsnames]{xcolor}
\usepackage{amsmath, amssymb, mathtools}  
\usepackage{bbding}
\usepackage{graphics}
\usepackage{wrapfig}
\usepackage{caption}
\usepackage{verbatim}
\usepackage{paralist}
\usepackage{cancel}
\usepackage{xspace}
\usepackage{multicol}
\usepackage{subcaption}
\usepackage{algorithm}
\usepackage{tabularx}
\usepackage{diagbox}
\usepackage{cite}
\usepackage{algorithmic}

\usepackage{tikz,ifthen}
\usetikzlibrary{arrows,trees,backgrounds,automata,shapes,decorations,plotmarks,fit,calc,positioning,shadows,chains}
\tikzstyle{every pin edge}=[<-,shorten <=1pt]
\tikzstyle{neuron}=[circle,fill=black!25,minimum size=17pt,inner sep=0pt]
\tikzstyle{input neuron}=[neuron, fill=green!40]
\tikzstyle{output neuron}=[neuron, fill=red!40]
\tikzstyle{hidden neuron}=[neuron, fill=blue!40]
\tikzstyle{constructed neuron}=[neuron, fill=orange!50]
\tikzstyle{annot} = [text width=6em, text centered]
\tikzstyle{nnedge} = [-{stealth},shorten >=0.1cm, shorten <=0.05cm,line width=0.8pt,black]
\usetikzlibrary{calc}
\usetikzlibrary{positioning}

\newcommand{\operator}[1]{\normalfont \texttt{#1}}
\newcommand{\certpred}{\operator{unsat}}
\newcommand{\abst}{\operator{abstract}}
\newcommand{\refine}{\operator{refine}}
\newcommand{\proveabst}{\operator{prove-over-approximation}}
\newcommand{\provequery}{\operator{verify-with-proofs}}
\newcommand{\verify}{\operator{verify}}

\newcommand{\interval}{\operator{bounds}}

\newcommand{\sat}{\operator{SAT}}
\newcommand{\unsat}{\operator{UNSAT}}

\newcommand{\x}{\mathbf{x}}
\newcommand{\h}{\mathbf{h}}
\newcommand{\y}{\mathbf{y}}
\newcommand{\nn}{f}
\newcommand{\examplenn}{\mathtt{f}}
\newcommand{\examplerednn}{\redVar{\examplenn}}

\newcommand{\numNeurons}{n}
\newcommand{\numLayers}{L}

\newcommand{\R}{\mathbb{R}}

\newcommand{\nnW}{\mathbf{W}}
\newcommand{\nnb}{\mathbf{b}}
\newcommand{\nnActFun}{\phi}

\newcommand{\bounds}{\mathcal{I}}
\newcommand{\bucket}{\mathcal{B}}
\newcommand{\redVar}[1]{\widehat{#1}}
\newcommand{\redNN}[0]{\redVar{\nn}}
\newcommand{\redBounds}[0]{\redVar{\bounds}}
\newcommand{\rednnW}{\redVar{\nnW}}
\newcommand{\rednnb}{\redVar{\nnb}}

\newcommand{\nnInputSet}{\mathcal{P}}
\newcommand{\nnHiddenSet}{\mathcal{H}}
\newcommand{\nnOutputSet}{\mathcal{Y}}
\newcommand{\exactSet}[1]{#1^*}

\newcommand{\nnHiddenSetExact}{\exactSet{\nnHiddenSet}}
\newcommand{\nnOutputSetExact}{\exactSet{\nnOutputSet}}

\newcommand{\nnHiddenSetRed}{\redVar{\nnHiddenSet}}
\newcommand{\nnOutputSetRed}{\redVar{\nnOutputSet}}

\newcommand{\drule}[2]{
	\renewcommand{\arraystretch}{1.2}
	\(\begin{array}{c}
		#1 \\
		\hline 
		#2
	\end{array}\)
}
\newcommand{\rulename}[1]{\ensuremath{\mathsf{#1}}\xspace}

\newcommand{\dnnvquery}[3]{\langle #1,\mathcal{#2},\mathcal{#3} \rangle}

\usepackage{hyperref}

\usepackage{cleveref}

\crefname{section}{Sec.}{Sec.}
\crefname{subsection}{Sec.}{Sec.}
\crefname{figure}{Fig.}{Fig.}
\crefname{algorithm}{Alg.}{Alg.}
\crefname{table}{Tab.}{Tab.}
\crefname{example}{Example}{Example}
\crefname{definition}{Def.}{Def.}
\crefname{proposition}{Prop.}{Prop.}
\crefname{theorem}{Thm.}{Thm.}
\crefname{lemma}{Lemma}{Lemmas}
\crefname{corollary}{Cor.}{Cor.}
\crefname{assumption}{Assumption}{Assumptions}
\crefname{appendix}{Appendix}{Appendix}
\crefname{equation}{Eq.}{Eq.}

\Crefname{section}{Sec.}{Sec.}
\Crefname{subsection}{Sec.}{Sec.}
\Crefname{figure}{Fig.}{Fig.}
\Crefname{algorithm}{Alg.}{Alg.}
\Crefname{table}{Tab.}{Tab.}
\crefname{example}{Example}{Example}
\Crefname{definition}{Def.}{Def.}
\Crefname{proposition}{Prop.}{Prop.}
\Crefname{theorem}{Thm.}{Thm.}
\Crefname{lemma}{Lemma}{Lemmas}
\Crefname{corollary}{Cor.}{Cor.}
\Crefname{assumption}{Assumption}{Assumptions}
\Crefname{appendix}{Appendix}{Appendix}
\Crefname{equation}{Eq.}{Eq.}

\title{Abstraction-Based Proof Production in \\Formal Verification of Neural Networks}
\titlerunning{Abtract Proof Production}
\author{
    Yizhak Yisrael Elboher\inst{1}\orcidID{0000-0003-2309-3505}\Envelope, Omri Isac\inst{1}\orcidID{0009-0000-6505-6900}, Guy Katz\inst{1}\orcidID{0000-0001-5292-801X}, Tobias Ladner\inst{2}\orcidID{0000-0002-4556-8308}, Haoze Wu\inst{3}\orcidID{0000-0002-5077-144X}
}
\authorrunning{YY. Elboher, O. Isac, G. Katz, T. Ladner, H. Wu}
\institute{ 
The Hebrew University of Jerusalem, Israel \and Technical University of Munich, Germany \and Amherst College, USA
}

\begin{document}

\maketitle

\begin{abstract}
Modern verification tools for deep neural networks (DNNs) increasingly rely on abstraction to scale to realistic architectures. In parallel, proof production is becoming a critical requirement for increasing the reliability of DNN verification results. However, current proof-producing verifiers do not support abstraction-based reasoning, creating a gap between scalability and provable guarantees.
We address this gap by introducing a novel framework for proof-producing abstraction-based DNN verification. Our approach modularly separates the verification task into two components: (i) proving the correctness of an abstract network, and (ii) proving the soundness of the abstraction with respect to the original DNN. 
The former can be handled by existing proof-producing verifiers, whereas we propose the first method for generating formal proofs for the latter.
This preliminary work aims to enable scalable and trustworthy verification by supporting common abstraction techniques within a formal proof framework. 
\end{abstract}

\section{Introduction}\label{sec:introduction}

Deep Neural Networks (DNNs) \cite{LeBeHi15, GoBeCo16} have demonstrated exceptional performance in various domains, including vision~\cite{KrSuHi12}, language~\cite{VaShPaUsJoGoKaPo17}, audio\cite{OoDiZeSiViGrKaSeKa16} and video\cite{ArDeHeSuLuSc21} analysis, achieving state-of-the-art accuracy in complex tasks \cite{ReKiXuBrMcSu23, RaWoHaRaGoAgSaAsMiClKrSu21}.
However, despite their success, DNNs function as black-box models, making their decision-making processes difficult to interpret and trust \cite{Ru19, Li18}.

DNN verification \cite{Eh17, KaBaDiJuKo17, LiArLaStBaKo21} provides formal methods and tools (\emph{verifiers}), to ensure or refute that DNNs comply with required specifications, offering formal guarantees of correctness. However, although verification algorithms are theoretically sound, their implementation occasionally introduce bugs and vulnerabilities~\cite{ZoBaCsIsJe21, JiRi21, ElKa23}, compromising their soundness and undermining the confidence in the verifier. 

A notable approach to tackle these issues is by producing \emph{formal proofs}, i.e., mathematical objects that can be checked by an independent program and witness the verifier's correctness. Proof production was explored in SMT and SAT solvers\cite{GrRoTo21,BaDeFo15}, and recently also in DNN verification~\cite{IsBaZhKa22,SiSaMeSi25}.
Although proofs enhance the reliability of the verification process, their generation limits the scalability of the verifier in two ways: \begin{inparaenum}[(a)]
    \item the generated proofs tend to be large, which
      substantially increases the verifier's memory consumption; and
    \item some verifier optimizations are not supported by the proof
      mechanism, and are disabled whenever proof generation is used
      --- slowing down the verifier. 
\end{inparaenum}

Scalability is a key challenge for DNN verification, which is an NP-complete problem~\cite{SaLa21} in simple cases, and modern solvers might solve verification queries in worst-case exponential time with respect to DNN size (number of neurons)\cite{BrMuBaJoLi23}. 
A common attempt to overcome this obstacle is to apply
\emph{abstraction}. This well-established technique in formal
verification\cite{CoCo77,ClGrLo94,ClGrJhLuVe00} is used to manage the
complexity of analyzing large systems by creating a simpler, abstract
model that retains the essential properties of the original system. In
the context of DNNs, abstraction has gained attention as a method to
enhance the scalability and efficiency of  verification\cite{ElGoKa22,AsHaKrMo20,LaAl23}. Specifically, DNN abstraction involves the construction of a reduced or approximate representation of the network such that the verification of the abstract network provides meaningful guarantees for the original network. 
Using abstraction, verification tools can handle larger networks and
more complex properties, making it a promising approach for scalable
and efficient formal analysis of DNNs.

This work-in-progress addresses two key challenges in DNN verification: enabling
proof production for abstraction-based solvers and generating more
compact proofs.  While abstraction improves scalability by simplifying
the network, existing proof-producing tools do not support it. To
bridge this gap, we propose the notion of an \emph{abstract
  proof} --- a modular proof consisting of (i) a proof that the required specification holds in the abstract network,
and (ii) a proof that the abstraction over-approximates
the original network, which means that if the property holds for the abstract network, it is guaranteed to hold for the original network as well.
  
Therefore, our approach extends proof support to scalable abstraction-based
solvers, while at the same time reducing proof size; since the abstract networks are typically smaller than the original DNNs (i.e., contain fewer neurons) and their verification time is faster, it is expected that size of the proof will be considerably reduced as well.
\Cref{fig:flowchart} illustrates the improved proof workflow and expected efficiency gains compared to the standard approach.
\begin{figure}
\centering
\begin{tikzpicture}[
  node distance=1cm and 0.75cm,
  box/.style={diamond, fill=blue!10, thick, minimum width=2cm, minimum height=1cm, align=center,
  rounded corners=2pt
  },
  operator/.style={rectangle, , fill=orange!40, thick, minimum width=1cm, minimum height=1cm, align=center,
  rounded corners=2pt
  },
  proofpart/.style={rectangle,, fill=yellow!40, thick, minimum width=1cm, minimum height=1cm, align=center,
  rounded corners=2pt
  },
  checkproof/.style={rectangle, fill=green!30, thick, minimum width=1cm, minimum height=1cm, align=center,
  rounded corners=2pt
  },
  standardproofpart/.style={rectangle, fill=yellow!20, thick, minimum width=1cm, minimum height=1cm, align=center,
  rounded corners=2pt
  },
  standardcheckproof/.style={rectangle, fill=green!10, thick, minimum width=1cm, minimum height=1cm, align=center,
  rounded corners=2pt
  }
  ]

  \node[box] (query) {Input\\Query};
  \node[operator, right=of query] (abstract) {Abstract};
  \node[proofpart, right=of abstract, yshift=0.5cm] (solve) {Verify \& \\ Prove};
  \node[proofpart, below=of solve, yshift=0.75cm]  (correctness) {Proof of\\Abstraction};
  \node[checkproof, right=2.9cm of abstract] (checkproof) {Check\\Proofs};

  \node[standardproofpart, left=of query] (standardsolve) {Verify \& \\ Prove};
  \node[standardcheckproof, left=of standardsolve] (standardcheckproof) {Check\\Proof};

  \draw[->] (query) -- node[midway, above, yshift=0.75cm, xshift=1cm]{\textbf{ours}}(abstract);
  \draw[->] (abstract) -- (solve);
  \draw[->] (abstract) -- (correctness);
  \draw[->] (solve) -- (checkproof);
  \draw[->] (correctness) -- (checkproof);
  \draw[->] (standardsolve) -- (standardcheckproof);
  \draw[->] (query) -- node[above, yshift=0.75cm, xshift=-1cm]{\textbf{standard}} (standardsolve);
  
\end{tikzpicture}
\caption{Proof production flowchart: standard (left) versus ours (right). Bold colors represent cheaper operations.} 
  \label{fig:flowchart}
\end{figure}
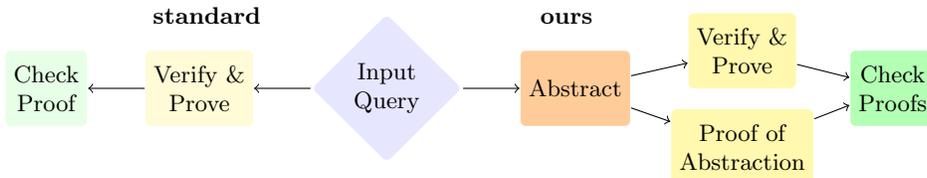

We regard this work as an attempt to lay a solid foundation, to be followed in the future by an  implementation and evaluation.

Inspired by CEGAR~\cite{ClGrJhLuVe00}, our main contributions are:
\begin{enumerate}
    \item We introduce the concept of an abstract proof for DNN verification.
    \item We design an abstraction-refinement mechanism for proof production.
    \item We formalize a verifiable proof of the abstraction process itself, using the Marabou DNN verifier~\cite{WuIsZeTaDaKoReAmJuBaHuLaWuZhKoKaBa24} and the CORA abstraction engine~\cite{Al15}.
\end{enumerate}

The paper is organized as follows:
\Cref{sec:preliminaries} provides background on DNNs, DNN verification and abstraction, and proof production. 
\Cref{sec:method} introduces our modular framework for constructing abstract proofs, describes the challenge of aligning proofs between the original and abstract networks, and presents a general abstraction-refinement algorithm for efficient proof production.
\Cref{sec:cora_abstraction} explains the abstraction process in CORA, and adapts it to our formulation.
\Cref{sec:implementation} details our implementation using Marabou (for proof production) and CORA (for abstraction), including how to verify abstract networks and generate corresponding proofs.
\Cref{sec:related-work} reviews relevant literature, and \Cref{sec:conclusion} summarizes our contributions and outlines future work.

\section{Preliminaries}\label{sec:preliminaries}

\subsection{Deep Neural Networks (DNNs)}
A \textit{Deep Neural Network (DNN)} is a parameterized function \( \nn\colon \R^{\numNeurons_0} \to \R^{\numNeurons_\numLayers}  \), composed of multiple layers of interconnected neurons. 
Each layer performs an affine transformation followed by a nonlinear activation function.
Formally, given an input $\x\in\R^{\numNeurons_0}$, the output $\y=\nn(\x)$ is computed as follows:
\begin{align}
\label{eq:dnn}
    \h_0 &= \x, & \h_{k} &= \nnActFun_{k}(\nnW_{k} \h_{k-1} + \nnb_{k}), & \y &= \h_\numLayers, & k\in[\numLayers].
\end{align}
where \( \nnW_{k}\in\R^{\numNeurons_{k}\times \numNeurons_{k-1}} \) is the weight matrix, \( \nnb_{k}\in\R^{\numNeurons_k} \) is the bias
vector, and \( \nnActFun_{k} \) is the nonlinear activation function
(e.g., ReLU\cite{NaHi10} or sigmoid) for the \( k \)-th layer. An
illustration of a neural network $\examplenn$ appears in \cref{fig:running-example}.

\thispagestyle{empty}
\definecolor{lightblue}{RGB}{100, 150, 200}
        
    \tikzset{
        every label/.style={font=\footnotesize, fill=white, inner sep=1pt, outer sep=0pt},
        neuron/.style={circle, fill=lightblue, draw=lightblue, thick, inner sep=0.08cm},
        weight/.style={fill=white, font=\tiny, inner sep=1pt, outer sep=0pt}
    }

\begin{figure}[ht]
  \centering
    \begin{tikzpicture}[scale=1]
        \footnotesize
    
        \newcommand{\inputx}{0}
        \newcommand{\hiddenx}{2.5}
        \newcommand{\outputx}{4.5}
        \newcommand{\hya}{4.0}
        \newcommand{\hyb}{3.0}
        \newcommand{\hyc}{2.0}
        \newcommand{\hyd}{1.0}
        \newcommand{\hye}{0.0}
        \newcommand{\inya}{4.0}
        \newcommand{\inyb}{2.66}
        \newcommand{\inyc}{1.33}
        \newcommand{\inyd}{0.0}
        \newcommand{\outya}{2.0}
    
        \foreach \i/\y in {1/\inya, 2/\inyb, 3/\inyc, 4/\inyd} {
            \foreach \j/\yy in {1/\hya, 2/\hyb, 3/\hyc, 4/\hyd, 5/\hye} {
                \draw (\inputx,\y) -- (\hiddenx,\yy);
            }
        }
        \foreach \i/\y in {1/\inya, 2/\inyb, 3/\inyc, 4/\inyd} {
            \foreach \j/\yy in {1/\hya, 2/\hyb, 3/\hyc, 4/\hyd, 5/\hye} {
                \pgfmathsetmacro\w{{
                {-1, -1, -2, 0},
                {-1, -2, -1, -3},
                {0.1, 0, 0, 0},
                {0, 0.2, 0, 0},
                {1, 1, 1, -1.3}
                }[\j-1][\i-1]}
                \pgfmathsetmacro\pos{0.2 + 0.6*(mod(\i,2))}
                \draw[draw=none] (\inputx,\y) -- (\hiddenx,\yy)
                node[weight,pos=\pos] {\pgfmathprintnumber[precision=2]{\w}};
            }
        }
    
        \foreach \i/\y in {1/\hya, 2/\hyb, 3/\hyc, 4/\hyd, 5/\hye} {
            \foreach \j/\yy in {1/\outya} {
                \pgfmathsetmacro\w{{{1,1,-5,-5,1}}[\j-1][\i-1]}
                \pgfmathsetmacro\pos{0.3 + 0*(mod(\i,2))}
                \draw (\hiddenx,\y) -- (\outputx,\yy)
                node[weight,pos=\pos] {\pgfmathprintnumber[precision=2]{\w}};
            }
        }
    
        \node[neuron,label={$1$}] at (\inputx,\inya) {};
        \node[neuron,label={$1$}] at (\inputx,\inyb) {};
        \node[neuron,label={$1$}] at (\inputx,\inyc) {};
        \node[neuron,label={$1$}] at (\inputx,\inyd) {};
    
        \node[neuron,label={\tiny{b=0}}] at (\hiddenx,\hya) {};
        \node[neuron,label={\tiny{b=0}}] at (\hiddenx,\hyb) {};
        \node[neuron,label={\tiny{b=0}}] at (\hiddenx,\hyc) {};
        \node[neuron,label={\tiny{b=0}}] at (\hiddenx,\hyd) {};
        \node[neuron,label={\tiny{b=0}}] at (\hiddenx,\hye) {};
        
        \node[neuron,label={$0.2$}] at (\outputx,\outya) {};
    
        \node at (\inputx,\hya+0.8) {\small Input};
        \node at (2.5,\hya+0.8) {\small Hidden};
        \node at (\outputx,\hya+0.8) {\small Output};
    
    \end{tikzpicture}
  \caption{A neural network $\examplenn$ for which $\examplenn(1,1,1,1) = 0.2$. All biases are 0.}
  \label{fig:running-example}
\end{figure}  


\subsection{DNN Verification}
A \textit{verification query} is a triplet $\dnnvquery{f}{P}{Q}$ where \( \nn\colon \R^{\numNeurons_0} \to \R^{\numNeurons_\numLayers} \) is a DNN, $\nnInputSet\subset\R^{\numNeurons_0}$ is an input property and $\mathcal{Q}\subset\R^{\numNeurons_\numLayers}$ is an output property. \textit{DNN Verification} aims to solve verification queries by deciding whether there exists an input satisfying $\mathcal{P}$ for which its output of $\nn$ satisfies $\mathcal{Q}$:
\begin{equation}
    \label{eq:dnn-verification}
    \exists \x. \ \x \in\nnInputSet \wedge \nn(\x)\in \mathcal{Q}.
  \end{equation}
Typically, \( \mathcal{Q}\) is a set that characterizes an undesired behavior, such as vulnerability of \(\nn\) to adversarial perturbations or danger conditions. If an input \(\x \in  \mathcal{P} \) is found whose output \(\nn(\x) \in  \mathcal{Q} \), we say the query is \sat{}, and $\x$ serves as a counterexample to the desired property. Otherwise, if no such input exists, we say the query is \unsat{}, and thus the desired property is valid. 
To ease notation, we also denote the latter case as $\certpred(\dnnvquery{\nn}{P}{Q})$. This can be captured as follows:
\begin{equation}
    \label{eq:unsat-dnn_vericiation}
    \certpred(\dnnvquery{\nn}{\nnInputSet}{Q}) \equiv \forall \x. \ \x \in\nnInputSet \implies \nn(\x) \notin \mathcal{Q}.
\end{equation}
For simplicity, we assume for this work that $\mathcal{P}$ is a
hyperrectangle, although our approach can be generalized to any closed
set $\nnInputSet$, e.g., by over-approximating  $\mathcal{P}$ with a bounding box thereof. 
An example of a verification query is $\langle\nn_1,\mathcal{P}_1,\mathcal{Q}_1\rangle$ where $\nn_1$ is the network in \cref{fig:running-example}, $\mathcal{P}_1=\{(1\pm\epsilon,1\pm\epsilon,1\pm\epsilon,1) \mid  \epsilon\in[0, 0.1]\}$ and $\mathcal{Q}_1=\mathbb{R}_{\le0}$.

\subsection{DNN Abstraction for Formal Verification}
\label{sec:abs}
To accelerate DNN verification, it is often beneficial to reduce the size of the network through abstraction. However, since the verification target is the original DNN, one must ensure that any conclusions drawn from the abstract (simplified) model also apply to the original. The over-approximation requirement is represented as follows, where $\nn,\redNN$ are the original and abstract networks, respectively:
\[
\certpred(\dnnvquery{\redNN}{\nnInputSet}{Q}) \implies \certpred(\dnnvquery{\nn}{\nnInputSet}{Q}).
\]
If the abstraction is too coarse and yields a \sat{} result with a spurious example, i.e. one that is not a counter example in the original model, the
abstract model is iteratively \textit{refined}
--- made into a more
precise, albeit a larger, over-approximation --- until the verification query can be resolved correctly.

Our framework is designed to be compatible with a wide range of abstraction methods. In this work, we focus on CORA~\cite{Al15}, a MATLAB toolbox used for formal verification of neural networks via reachability analysis. With its recent abstraction-refinement extension~\cite{LaAl23}, CORA improves performance by replacing the original network with a smaller abstract model that is iteratively refined as needed. We leverage CORA as a backend for abstraction in our framework. Other approaches to network abstraction are discussed in \cref{sec:related-work}. In some abstraction methods, including the one used in CORA, the abstract network does not follow a standard neural network structure. \Cref{sec:fullproof} discusses this gap in more detail.

\subsection{Proof Production for DNN Verification} \label{sec:proof-production}
As a satisfiability problem, proving that a  DNN verification query is
 \sat{} is straightforward, using a satisfying assignment that could be checked by evaluation over the network.  
Proving \unsat{}, however, is more complicated due to the NP-hardness of DNN verification~\cite{KaBaDiJuKo21, SaLa21}. Thus, bookkeeping the whole proof may require large memory consumption, even for small DNNs. 

In this work, we focus on the proof producing version of Marabou~\cite{WuIsZeTaDaKoReAmJuBaHuLaWuZhKoKaBa24, IsBaZhKa22}, a state-of-the-art DNN verifier, which encodes verification queries as satisfiability problems, utilizing satisfiability modulo theories (SMT) solving and linear programming (LP) to analyze properties of interest. It handles nonlinear activation functions, such as ReLU, through case-splitting and relaxation techniques. 
Marabou's proof of \unsat{} is represented by a \emph{proof-tree}.  By
construction, the proof's size
heavily depends on the number of splits performed by Marabou, which
could be exponential in the number of neurons.

\section{Method}\label{sec:method}
We intend to accelerate verification by applying abstraction to reduce the DNN size and prove the property over the abstract DNN. However, a proof over the abstract network alone is insufficient --- it does not guarantee that the property holds for the original network. To overcome this, we introduce a general framework for constructing end-to-end proofs that remain sound while leveraging abstraction.

\subsection{Proving Abstraction-Based DNN Verification}
\label{sec:fullproof}
The verification of DNNs using abstraction consists of two main
components: constructing the abstraction and verifying the abstract
network. To match this structure, our proof method for \unsat{} cases follows the same modular approach. This modularity ensures that our method remains agnostic to the underlying DNN verifier and abstraction technique, making it broadly applicable. Specifically, it enables combining any DNN verification tool capable of producing proofs with any abstraction method that comes with a corresponding proof rule.  

Our method constructs a proof that consists of
two independent parts. First, given a candidate DNN $\nn$, an abstract netrowk  $\redNN$ and
properties $\mathcal{P}, \mathcal{Q}$, we establish that if
$\dnnvquery{\redNN}{P}{Q}$ is \unsat{}, then so is
$\dnnvquery{\nn}{P}{Q}$. This forms the \emph{proof of over-approximation}, i.e.,
proof of abstraction correctness, which ensures that verification
results transfer from the abstract network to the original one. 
The second part is the verification proof for the abstract network, i.e., the proof that $\dnnvquery{\redNN}{P}{Q}$ is indeed \unsat{}. These two proofs, when combined, yield the \emph{abstract proof} following the proof rule in~\Cref{fig:absproofrule}. Even though this rule is a private case of implication elimination (modus ponens), we define it to clearly indicate our modular approach. 

\begin{figure*}[h!]
    \centering
    \rulename{abs-proof:}
    \drule{
        \certpred(\dnnvquery{\redNN}{\nnInputSet}{Q}) \implies \certpred(\dnnvquery{\nn}{\nnInputSet}{Q})
        \ \ \ \
         \certpred(\dnnvquery{\redNN}{P}{Q})
        }
        {
        \certpred(\dnnvquery{\nn}{P}{Q})
        }

        \caption{Proof rule for proving DNN verification with abstraction.}
    \label{fig:absproofrule}
\end{figure*}

As a proof system for verifying DNNs (i.e., the top right part of the
rule) has been introduced in prior work~\cite{IsBaZhKa22}, we focus
on constructing the proof of over-approximation (i.e., the top left part). 
We exemplify this in~\Cref{sec:coraproofrule}. Also note that in the cases where $\redNN$ is not precisely a DNN, we should also formalize $\dnnvquery{\redNN}{P}{Q}$ and its unsatisfiability, as part of the abstraction's definition. Furthermore, we are required to show how  $\dnnvquery{\redNN}{P}{Q}$ can be reduced to a DNN verification query.  We do so for CORA  in~\Cref{sec:cora_abstraction} and in ~\Cref{sec:apply_marabou_on_abstraction}, respectively.

\subsection{Main Algorithm}
We propose \cref{alg:verify-then-prove-once} to improve proof production via abstraction-refinement. The algorithm begins (line 1) by generating an abstract version of the original network, then iteratively verifies the correctness of the desired property on the abstract network.

Unless the condition in line 6 is met, only verification is performed; proof production is attempted only after an \unsat{} result has been established. This avoids redundant proof attempts and improves performance in each iteration. An additional advantage is modularity: any verifier can be used to check the property, not just those with proof production capabilities or specific configurations that support it.

The procedures $\proveabst$ and $\provequery$ represent the generation of a proof for the over-approximation and for the abstract query, respectively, and are described in more detail in \cref{sec:implementation}. The pair $\langle p_a, p_q \rangle$ denotes the concatenation of the two components into a full proof.

\begin{algorithm}
\textbf{Input:} $\nn$, $\mathcal{P}$, $\mathcal{Q}$.
\textbf{Output:} proof that $\certpred(\dnnvquery{\nn}{\nnInputSet}{Q})$, or counterexample.
\caption{Proof Production with Abstraction}\label{alg:verify-then-prove-once}
\begin{algorithmic}[1]
        \STATE {$\redNN$ = $\abst(\nn, \mathcal{P})$}
        \WHILE{ true }
            \STATE {result, example = $\verify(\redNN$,$\mathcal{P}, \mathcal{Q}$)}
            \IF{result == \sat{} and example is not spurious}
                \RETURN {result, example}
            \ELSIF{result == \unsat{}}
                \STATE{$p_a$ = $\proveabst$($\redNN$,$\nn,\nnInputSet,\mathcal{Q}$)}
                \STATE{$p_q$ = $\provequery$($\redNN$,$\mathcal{P}$,$\mathcal{Q}$)}
                \IF{$p_a$ and $p_q$ were successfully generated}
                \RETURN{\unsat{}, $\langle p_a, p_q \rangle$}
            \ENDIF
            \ENDIF
            \STATE {$\redNN$ = $\refine(\redNN,\nn)$} 
        \ENDWHILE
\end{algorithmic}
\end{algorithm}


\section{Abstraction in CORA}\label{sec:cora_abstraction}
We provide an overview on how abstraction in CORA works and how it can be integrated into the
verification process. We refer the reader to \cite{Al15,LaAl23} for additional details.

 Given a neural network $\nn$ as in~\eqref{eq:dnn} and an input set $\nnInputSet$, the exact output set $\nnOutputSetExact=\nn(\nnInputSet)$ is computed by
 \begin{align}
     \label{eq:dnn-set-propagation-exact}
     \nnHiddenSetExact_0 &= \nnInputSet, & \nnHiddenSetExact_{k} &= \nnActFun_{k}(\nnW_{k} \nnHiddenSetExact_{k-1} + \nnb_{k}), & \nnOutputSetExact &= \nnHiddenSetExact_\numLayers, & k\in[\numLayers].
 \end{align}
 These exact sets are generally expensive to compute~\cite{KaBaDiJuKo17}.
 Thus, we over approximate the output of each layer $\nnHiddenSet_k \supseteq \nnHiddenSet^*_k$.
 In this work, we only consider the set $\nnHiddenSet_k$ to be
 represented as hyperractangles,
 although more sophisticated set representations exist~\cite{GeMiDrTsChVe18,bak2021nnenum,lopez2023nnv,kochdumper2022open,ladner2023automatic}.

Since DNNs usually contain a large number of neurons per layer, their  verification can be computationally expensive as well.
Thus, \cite{LaAl23} suggests a construction of an abstract network that soundly merges neurons with similar bounds to reduce the network size, 
which in turn decreases the verification time by decreasing the computation time.
The bounds are determined by a one-step look-ahead algorithm using
interval bound propagation (IBP)~\cite{GoDvStBuQiUeArMaKo19}.
In particular, we compute the output interval bounds of layer $k$ as follows \cite[Alg.~2]{LaAl23}:
\begin{equation}
    \label{eq:bounds-computation}
    \bounds_k = \nnActFun_k(\nnW_{k}\cdot\interval(\nnHiddenSet_{k-1})) + \nnb_{k}) \supseteq \nnHiddenSet_k,
\end{equation}
where $\interval(\cdot)$ computes the interval bounds of the given set $\nnHiddenSet_{k-1}$.
In order to preserve soundness, 
multiple neurons with similar bounds are merged and the resulting error is bounded by adapting the bias term in the next layer, 
converting them from scalars into intervals. 
These bias intervals bound the deviation between the abstract network and the original network of each layer in the network and, thus, also the output of both networks.

More formally, given a neural network, 
\cite[Prop.~4]{LaAl23} defines a way to merge the neurons in the $k$-th layer, 
constructing the weights and biases such that the output of the $k+1$-th layer of the original neural network is contained in the output of the $k+1$-th layer of the abstract network.
Notice that the terminology in \cite{LaAl23} splits each layer into
two layers, namely the linear layer and the nonlinear layer, and indexes them separately. 
Here, we similarly treat each layer as having two parts, but do not handle these as different layers.

\begin{proposition}[{Neuron Merging \cite[Prop.~4]{LaAl23}}]
    \label{prop:neuron-merging}
    Given a nonlinear hidden layer $k\in [\numLayers-1]$ of a network $\nn$ with $\numNeurons_k$ neurons, 
    output interval bounds $\bounds_k\supseteq \nnHiddenSet^*_k$, a
    merge bucket $\bucket\subset[\numNeurons_k]$ containing the
    indices of the merged neurons, and
    $\bar{\bucket}=[\numNeurons_k]\setminus\bucket$, we can construct
    an abstract network $\redNN$,
    where we remove the merged neurons by adjusting the layers $k$ and $k + 1$ as follows:
    \begin{align*}
        \rednnW_{k} &= \nnW_{k(\bar{\bucket},\cdot)}, &
        \rednnb_{k} &= \nnb_{k(\bar{\bucket})}, &
        \redBounds_{k} &= \redBounds_{k(\bar{\bucket})}, \\
        \rednnW_{k+1} &= \nnW_{k+1(\cdot, \bar{\bucket})}, & 
        \rednnb_{k+1} &= \nnb_{k+1}, &
        \redBounds_{k+1} &= \nnW_{k+1(\cdot,\bucket)}\bounds_{k(\bucket)}.
    \end{align*}
\end{proposition}
$\newline$
where for $\mathcal{S}\in\{\bucket,\bar{\bucket}\}$, $\square_{(\mathcal{S},\cdot)}$ and $\square_{(\cdot, \mathcal{S})}$ represent the rows and columns with the indices in $\mathcal{S}$ in lexicographic order, respectively.
The interval bounds $\redBounds_k$ require us to extend the formulation for a neural network~$\nn$ as in~\eqref{eq:dnn} to an abstract network $\redNN$. \cref{fig:basic_abstract_network} illustrates this extension for the neural network $\examplenn$.
\thispagestyle{empty}
\definecolor{lightblue}{RGB}{100, 150, 200}        
    \tikzset{
        every label/.style={font=\footnotesize, fill=white, inner sep=1pt, outer sep=0pt},
        neuron/.style={circle, fill=lightblue, draw=lightblue, thick, inner sep=0.08cm},
        weight/.style={fill=white, font=\tiny, inner sep=1pt, outer sep=0pt}
    }
\begin{figure}[ht]
  \centering
    \begin{tikzpicture}
        \newcommand{\inputx}{0}
        \newcommand{\hiddenx}{2.5}
        \newcommand{\outputx}{4.5}
        \newcommand{\hya}{4.0}
        \newcommand{\hyb}{3.0}
        \newcommand{\hyc}{2.0}
        \newcommand{\hyd}{1.0}
        \newcommand{\hye}{0.0}
        \newcommand{\inya}{4.0}
        \newcommand{\inyb}{2.66}
        \newcommand{\inyc}{1.33}
        \newcommand{\inyd}{0.0}
        \newcommand{\outya}{2.0}
    
        \foreach \i/\y in {1/\inya, 2/\inyb, 3/\inyc, 4/\inyd} {
            \foreach \j/\yy in {1/\hya, 2/\hyb, 3/\hyc, 4/\hyd, 5/\hye} {
                \draw (\inputx,\y) -- (\hiddenx,\yy);
            }
        }
        \foreach \i/\y in {1/\inya, 2/\inyb, 3/\inyc, 4/\inyd} {
            \foreach \j/\yy in {1/\hya, 2/\hyb, 3/\hyc, 4/\hyd, 5/\hye} {
                \pgfmathsetmacro\w{{
                {-1, -1, -2, 0},
                {-1, -2, -1, -3},
                {0.1, 0, 0, 0},
                {0, 0.2, 0, 0},
                {1, 1, 1, -1.3}
                }[\j-1][\i-1]}
                \pgfmathsetmacro\pos{0.2 + 0.6*(mod(\i,2))}
                \draw[draw=none] (\inputx,\y) -- (\hiddenx,\yy)
                node[weight,pos=\pos] {\pgfmathprintnumber[precision=2]{\w}};
            }
        }
    
        \foreach \i/\y in {1/\hya, 2/\hyb, 3/\hyc, 4/\hyd, 5/\hye} {
            \foreach \j/\yy in {1/\outya} {
                \pgfmathsetmacro\w{{{1,1,-5,-5,1}}[\j-1][\i-1]}
                \pgfmathsetmacro\pos{0.3 + 0*(mod(\i,2))}
                \draw (\hiddenx,\y) -- (\outputx,\yy)
                node[weight,pos=\pos] {\pgfmathprintnumber[precision=2]{\w}};
            }
        }
    
        \node[neuron] at (\inputx,\inya) {};
        \node[neuron] at (\inputx,\inyb) {};
        \node[neuron] at (\inputx,\inyc) {};
        \node[neuron] at (\inputx,\inyd) {};
    
        \node[neuron,label={\tiny{$[0,0]$}}] at (\hiddenx,\hya) {};
        \node[neuron,label={\tiny{$[0,0]$}}] at (\hiddenx,\hyb) {};
        \node[neuron,label={\tiny{$[0,0]$}}] at (\hiddenx,\hyc) {};
        \node[neuron,label={\tiny{$[0,0]$}}] at (\hiddenx,\hyd) {};
        \node[neuron,label={\tiny{$[0,0]$}}] at (\hiddenx,\hye) {};
        
        \node[neuron] at (\outputx,\outya) {};
    
        \node at (\inputx,\hya+0.8) {\small Input};
        \node at (2.5,\hya+0.8) {\small Hidden};
        \node at (\outputx,\hya+0.8) {\small Output};
    
    \end{tikzpicture}
    \caption{$\examplerednn_0$, the extension of $\examplenn$ with the new formulation for abstract networks, where the biases of $\examplenn$ (zeros) are converted to intervals (singletones).}
    \label{fig:basic_abstract_network}
\end{figure}

Given an input $\x\in\nnInputSet$, the output $\nnOutputSetRed=\redNN(\x)$ is computed by
\begin{align}
\label{eq:rednn}
    \nnHiddenSetRed_0 &= \{\x\}, & \nnHiddenSetRed_k &= \nnActFun_k(\rednnW_k \nnHiddenSetRed_{k-1} \oplus \rednnb_k \oplus \redBounds_k), & \nnOutputSetRed &= \nnHiddenSetRed_\numLayers, & k\in[\numLayers],
\end{align}
where all $\redBounds_k$ are initialized with $\{\mathbf{0}\}$, or equivalently $[\mathbf{0},\mathbf{0}]$,
 and $\oplus$ denotes the Minkowski sum of two sets, i.e., given $\mathcal{S}_1,\,\mathcal{S}_2\subset\R^n$, $\mathcal{S}_1\oplus\mathcal{S}_2=\{s_1+s_2\ |\ s_1\in\mathcal{S}_1,\, s_2\in\mathcal{S}_1\}$.
If either summand of the Minkowski sum is given as a vector, it is implicitly converted to a singleton.
The interval biases $\rednnb_k\oplus\redBounds_k$ capture the error between the abstract network and the original network.
Thus, initially it holds that:
\begin{align}
\label{eq:corabasecase}
    \forall \x\in\nnInputSet\colon \nn(\x) \in \redNN(\x) = \{\nn(\x)\}.
\end{align}
As an abstract network outputs a set instead of a single vector, we also have to generalize~\eqref{eq:unsat-dnn_vericiation} to abstract networks:
\begin{equation}
    \label{eq:unsat-rednn_vericiation}
    \certpred(\dnnvquery{\redNN}{\nnInputSet}{Q}) \equiv \forall \x. \ \x \in\nnInputSet \implies \redNN(\x) \cap \mathcal{Q} = \emptyset.
\end{equation}
This formulation enables us the following corollary:
\begin{corollary}
\label{lem:abstraction-correctnes-conditions}
Given an input set $\nnInputSet$,
an abstract neural network $\redNN$ where \cref{prop:neuron-merging} is applied to all layers $k' \leq k\in[\numLayers]$,
we can merge the neurons in layer $k$ using \cref{prop:neuron-merging} such that for the obtained abstract network $\redNN'$, it holds that:
\begin{equation*}
    \forall\x\in\nnInputSet\colon  \redNN(\x) \subseteq \redNN'(\x).
\end{equation*}
In particular, it holds that:
\begin{equation*}
    \certpred(\dnnvquery{\redNN'}{\nnInputSet}{Q}) \implies \certpred(\dnnvquery{\redNN}{\nnInputSet}{Q}).
\end{equation*}
\end{corollary}
\begin{proof}
    Let $\nnHiddenSetRed_k$ and $\nnHiddenSetRed'_k$ denote the output of the $k$-th layer of $\redNN$ and $\redNN'$, respectively.
 Note that \cref{prop:neuron-merging} only alters layer $k$ and $k+1$, thus, all other layers are identical between $\redNN$ and $\redNN'$~\eqref{eq:rednn}.
    In particular, we know that $\nnHiddenSetRed_{k-1} = \nnHiddenSetRed'_{k-1}$ holds.
    We now show that $\nnHiddenSetRed_{k+1} \subseteq \nnHiddenSetRed'_{k+1}$ holds by contradiction.
    Let us assume that there exist a $\bar{\h}_{k-1}\in\nnHiddenSetRed_{k-1}$ for which the respective $\bar{\h}_{k+1}\in\nnHiddenSetRed_{k+1}$ but $\bar{\h}_{k+1}\not\in\nnHiddenSetRed'_{k+1}$.
    However, this cannot be true as the values of the merged neurons are captured by $\redBounds_{k+1}$, which is computed over-approximative using IBP (\eqref{eq:bounds-computation}, \cref{prop:neuron-merging}),
    and all remaining neurons are kept equal (\cref{prop:neuron-merging}).
    Thus, $\nnHiddenSetRed_{k+1} \subseteq \nnHiddenSetRed'_{k+1}$ holds, which directly shows that $ \redNN(\x) \subseteq \redNN'(\x)$ as all subsequent layers are again identical.
    The implication directly follows due to the subset relation and~\eqref{eq:unsat-rednn_vericiation}. $\qed$
\end{proof}
  
\thispagestyle{empty}
\definecolor{lightblue}{RGB}{100, 150, 200}
        
    \tikzset{
        every label/.style={font=\footnotesize, fill=white, inner sep=1pt, outer sep=0pt},
        neuron/.style={circle, fill=lightblue, draw=lightblue, thick, inner sep=0.08cm},
        weight/.style={fill=white, font=\tiny, inner sep=1pt, outer sep=0pt},
        abstractneuron/.style={circle, fill=white, draw=lightblue, thick, inner sep=0.08cm},
        weight/.style={fill=white, font=\tiny, inner sep=1pt, outer sep=0pt},
    }

\begin{figure}[ht]
  \centering
  \begin{minipage}[b]{0.45\textwidth}
    \centering
    \begin{tikzpicture}[scale=1]
        \footnotesize
    
        \newcommand{\inputx}{0}
        \newcommand{\hiddenx}{2.5}
        \newcommand{\outputx}{4.5}
        \newcommand{\hya}{4.0}
        \newcommand{\hyb}{3.0}
        \newcommand{\hyc}{2.0}
        \newcommand{\hyd}{1.0}
        \newcommand{\hye}{0.0}
        \newcommand{\inya}{4.0}
        \newcommand{\inyb}{2.66}
        \newcommand{\inyc}{1.33}
        \newcommand{\inyd}{0.0}
        \newcommand{\outya}{2.0}
    
        \foreach \i/\y in {1/\inya, 2/\inyb, 3/\inyc, 4/\inyd} {
            \foreach \j/\yy in {1/\hya, 2/\hyb, 3/\hyc, 4/\hyd, 5/\hye} {
                \draw (\inputx,\y) -- (\hiddenx,\yy);
            }
        }
        \foreach \i/\y in {1/\inya, 2/\inyb, 3/\inyc, 4/\inyd} {
            \foreach \j/\yy in {1/\hya, 2/\hyb, 3/\hyc, 4/\hyd, 5/\hye} {
                \pgfmathsetmacro\w{{
                {-1, -1, -2, 0},
                {-1, -2, -1, -3},
                {0.1, 0, 0, 0},
                {0, 0.2, 0, 0},
                {1, 1, 1, -1.3}
                }[\j-1][\i-1]}
                \pgfmathsetmacro\pos{0.2 + 0.6*(mod(\i,2))}
                \draw[draw=none] (\inputx,\y) -- (\hiddenx,\yy)
                node[weight,pos=\pos] {\pgfmathprintnumber[precision=2]{\w}};
            }
        }
    
        \foreach \i/\y in {1/\hya, 2/\hyb, 3/\hyc, 4/\hyd, 5/\hye} {
            \foreach \j/\yy in {1/\outya} {
                \pgfmathsetmacro\w{{{1,1,-5,-5,1}}[\j-1][\i-1]}
                \pgfmathsetmacro\pos{0.5 + 0*(mod(\i,2))}
                \draw (\hiddenx,\y) -- (\outputx,\yy)
                node[weight,pos=\pos] {\pgfmathprintnumber[precision=2]{\w}};
            }
        }
    
        \node[neuron,label={\tiny{[0.9,1.1]}}] at (\inputx,\inya) {};
        \node[neuron,label={\tiny{[0.9,1.1]}}] at (\inputx,\inyb) {};
        \node[neuron,label={\tiny{[0.9,1.1]}}] at (\inputx,\inyc) {};
        \node[neuron,label={\tiny{[1,1]}}] at (\inputx,\inyd) {};
    
        \node[neuron,label={[xshift=5pt, yshift=2pt]\tiny{[0,0]}}] at (\hiddenx,\hya) {};
        \node[neuron,label={[xshift=5pt, yshift=2pt]\tiny{[0,0]}}] at (\hiddenx,\hyb) {};
        \node[neuron,label={[xshift=7pt, yshift=2pt]\tiny{[0.09,0.11]}}] at (\hiddenx,\hyc) {};
        \node[neuron,label={[xshift=7pt, yshift=2pt]\tiny{[0.18,0.22]}}] at (\hiddenx,\hyd) {};
        \node[neuron,label={[xshift=7pt, yshift=2pt]\tiny{[1.4,2]}}] at (\hiddenx,\hye) {};
        
        \node[neuron,label={\tiny{[0.05,0.35]}}] at (\outputx,\outya) {};
    
        \node at (\inputx,\hya+0.8) {\small Input};
        \node at (2.5,\hya+0.8) {\small Hidden};
        \node at (\outputx,\hya+0.8) {\small Output};
    
    \end{tikzpicture}
    \caption*{(a) Basic abstract network $\examplerednn_1$}
  \end{minipage}
  \hfill
  \begin{minipage}[b]{0.45\textwidth}
    \centering
    \begin{tikzpicture}
        \newcommand{\inputx}{0}
        \newcommand{\hiddenx}{2.5}
        \newcommand{\outputx}{4.5}
        \newcommand{\hya}{4.0}
        \newcommand{\hyb}{3.0}
        \newcommand{\hyc}{2.0}
        \newcommand{\hyd}{1.0}
        \newcommand{\hye}{0.0}
        \newcommand{\inya}{4.0}
        \newcommand{\inyb}{2.66}
        \newcommand{\inyc}{1.33}
        \newcommand{\inyd}{0.0}
        \newcommand{\outya}{2.0}
    
        \foreach \i/\y in {1/\inya, 2/\inyb, 3/\inyc, 4/\inyd} {
            \foreach \j/\yy in {1/\hyd, 2/\hye} {
                \draw (\inputx,\y) -- (\hiddenx,\yy);
            }
        }
        \foreach \i/\y in {1/\inya, 2/\inyb, 3/\inyc, 4/\inyd} {
            \foreach \j/\yy in {1/\hyd, 2/\hye} {
                \pgfmathsetmacro\w{{
                {0, 0.2, 0, 0},
                {1, 1, 1, -1.3}
                }[\j-1][\i-1]}
                \pgfmathsetmacro\pos{0.25 + 0*(mod(\i,2))}
                \draw[draw=none] (\inputx,\y) -- (\hiddenx,\yy)
                node[weight,pos=\pos] {\pgfmathprintnumber[precision=2]{\w}};
            }
        }
    
        \foreach \i/\y in {1/\hyd, 2/\hye} {
            \foreach \j/\yy in {1/\outya} {
                \pgfmathsetmacro\w{{{-5,1}}[\j-1][\i-1]}
                \pgfmathsetmacro\pos{0.5 + 0*(mod(\i,2))}
                \draw (\hiddenx,\y) -- (\outputx,\yy)
                node[weight,pos=\pos] {\pgfmathprintnumber[precision=2]{\w}};
            }
        }

        \draw[lightblue] (\hiddenx,\hyb) -- (\outputx,\outya) node[weight,pos=0.4] {\tiny{[-0.55,-0.45]} $\oplus$ \tiny{[0,0]}};
    
        \node[neuron,label={\tiny{[0.9,1.1]}}] at (\inputx,\inya) {};
        \node[neuron,label={\tiny{[0.9,1.1]}}] at (\inputx,\inyb) {};
        \node[neuron,label={\tiny{[0.9,1.1]}}] at (\inputx,\inyc) {};
        \node[neuron,label={\tiny{[1,1]}}] at (\inputx,\inyd) {};
    
        \node[abstractneuron,label={[xshift=5pt, yshift=0pt]\tiny{$\mathcal{B}=\{1,2,3\}$}}] at (\hiddenx,\hyb) {};
        \node[neuron,label={[xshift=7pt, yshift=2pt]\tiny{[0.18,0.22]}}] at (\hiddenx,\hyd) {};
        \node[neuron,label={[xshift=7pt, yshift=2pt]\tiny{[1.4,2]}}] at (\hiddenx,\hye) {};
        
        \node[neuron,label={\tiny{[-0.05,0.45]}}] at (\outputx,\outya) {};
    
        \node at (\inputx,\hya+0.8) {\small Input};
        \node at (2.5,\hya+0.8) {\small Hidden};
        \node at (\outputx,\hya+0.8) {\small Output};
    
    \end{tikzpicture}
    \caption*{(b) After CORA-abstraction $\examplerednn$}
  \end{minipage}
  \caption{Example of abstraction, given the input property $\nnInputSet_1$. The basic abstract network $\examplerednn_1$ (left) is reduced to another abstract network $\examplerednn$ (right).}
  \label{fig:abstraction-example}
\end{figure}  

\noindent \textbf{Example:} An example of abstraction for  $\examplenn$ is shown in \cref{fig:abstraction-example}. Given the input set $\nnInputSet_1$, the output bounds of the abstract network $\examplerednn_1$ contain the output bounds of the basic abstract network $\examplerednn_0$. The first three hidden neurons have similar bounds, making them a merge bucket $\mathcal{B}=\{1,2,3\}$, and are thus merged into an abstract neuron (in white). The set of bounds of the bucket ([0,0], [0,0] and [0.09,0.11]) is embedded into the bias ([0,0]) of the neuron in the output layer using IBP ($1\cdot[0,0]+1\cdot[0,0]-5\cdot[0.09,0.11]=[-0.55,-0.45]$) and Minkowski sum ($\oplus$), resulting with output bounds of $[-0.05,0.45]=-5\cdot[0.18,0.22]+1\cdot[1.4,2]+[-0.55,-0.45]$. This results in the final abstract network $\examplerednn$.

\section{Proving CORA Abstraction and Marabou Verification}\label{sec:implementation}

In this work, we focus on the proof-producing version of Marabou~\cite{WuIsZeTaDaKoReAmJuBaHuLaWuZhKoKaBa24,IsBaZhKa22}, a state-of-the-art verification tool with the ability to produce proofs of its \unsat{} results. 
The abstraction process used in this work was suggested in \cite{LaAl23} and is
implemented to improve set-based DNN verification and is part of CORA\cite{Al15}.
In this section, we explain how to implement $\provequery$ and
$\proveabst$ in \cref{alg:verify-then-prove-once} with these tools.

We start with explaining (in \cref{sec:apply_marabou_on_abstraction}) the details about the verification process in Marabou, and then show how to implement $\provequery$ and apply Marabou on an abstract network obtained by the abstraction process in CORA. In \cref{sec:coraproofrule}, we show how to implement $\proveabst$ and produce a proof that the abstraction process is correct. By doing so, we accomplish both necessary results as outlined in~\Cref{sec:fullproof}.

\subsection{Proving Correctness of Abstract Network Queries in Marabou}\label{sec:apply_marabou_on_abstraction}
The abstract networks obtained by the abstraction process in CORA generalize DNNs. 
As a result, the verification process should be adapted from solving $\dnnvquery{\nn}{P}{Q}$ and proving \eqref{eq:unsat-dnn_vericiation} to solving $\dnnvquery{\redNN}{P}{Q}$ and proving \eqref{eq:unsat-rednn_vericiation}. In the following, we show how the latter can be represented as a query that the Marabou verifier supports. This allows implementing $\provequery$, since the proofs generated by Marabou during the verification of $\dnnvquery{\redNN}{P}{Q}$ can be used as proofs for \eqref{eq:unsat-rednn_vericiation}.
Recall that Marabou  handles verification queries by trying to find satisfying assignments to the linear constraints in $\nn$ with LP methods\cite{Ch83,Gu24} and check the solution against the other, non-linear, constraints. 

There are two differences to consider when using Marabou to solve
queries over abstract networks. First, the output of an abstract
network is a set of vectors and not a single vector. This does not require any change in the verification query, as Marabou is capable of handling constraints that define continuous sets in both input and output.
Second, the abstract network $\redNN$ structure expresses biases as intervals or singletons, instead of scalars. To address this, we propose two options for encoding the verification query in a form supported by Marabou:

\begin{enumerate}
	\item 
	The architecture of the original network $\nn$ is encoded in Marabou using equations. Since the backend LP solver used in Marabou  supports linear inequalities, linear constraints in an abstract network are represented directly as inequalities, where the lower and upper bounds reflect the abstracted bias interval.
	More formally, suppose the bias term $\widehat{b}_{ki}$ in the abstract network~$\redNN$ lies in the interval $[\widehat{b}_{ki}^l, \widehat{b}_{ki}^u]$. We can then express the linear constraint as the following pair of inequalities:
	\[
	n_{ki} \ge \sum_j \nnW_{ji}^{k-1} x_j^{k-1} + \widehat{b}_{ki}^l, \quad \text{and} \quad
	n_{ki} \le \sum_j \nnW_{ji}^{k-1} x_j^{k-1} + \widehat{b}_{ki}^u.
	\]
	
	As an example, consider the output neuron in the network $\nn_1$, which is originally encoded by the equation:
	\[
	n_{\text{out}} = \sum_{i=1}^5 \nnW_{i1}^{1} n_{1i} + 0.0,
	\]
	where $n_{1i}$ denotes the $i$-th neuron in the hidden layer, and the final term is the bias. After converting $\nn_1$ into its abstract counterpart $\redNN_1$, the output neuron is instead represented using the pair of inequalities:
	\[
	n_{\text{out}} \ge \sum_{i=1}^5 \nnW_{i1}^{1} n_{1i} + 0.0, \quad \text{and} \quad n_{\text{out}} \le \sum_{i=1}^5 \nnW_{i1}^{1} n_{1i} + 0.0.
	\]
	since its bias lies in the interval $[0, 0]$. After applying further abstraction to obtain $\redNN'_1$, these inequalities are updated to reflect the new bounds:
	\[
	n_{\text{out}} \ge \sum_{i=1}^5 \nnW_{i1}^{1} n_{1i} - 0.05, \quad \text{and} \quad n_{\text{out}} \le \sum_{i=1}^5 \nnW_{i1}^{1} n_{1i} + 0.45.
	\]
	\item The Marabou solver supports skip connections, as these can be represented as additional linear constraints, which are seamlessly handled by the underlying LP solver. Consequently, for each bias term $\nnb_{ki}$ in an abstract network $\redNN$, we can introduce a fresh input variable $p_{ki}$ that is connected via a skip connection of weight 1 directly to the neuron $n_{ki}$. The set $\mathcal{P}$ is then extended with a constraint that enforces the interval bounds associated with $p_{ki}$.
	
	For example, in the network $\redNN'_1$, the bias of the output neuron is encoded through an additional input variable $p_{\text{out}}$, and $\mathcal{P}_1$ is updated to include the constraint $-0.05 \le p_{\text{out}} \le 0.45$.
\end{enumerate}

Both methods allow us to verify $\dnnvquery{\redNN}{P}{Q}$ directly
with Marabou; the first introduces additional inequalities, whereas the second requires a larger number of variables.

\subsection{Proving Correctness of the CORA Abstraction}
\label{sec:coraproofrule}
After explaining the implementation of $\provequery$ in the previous section, we are left with showing how $\proveabst$  for the CORA abstraction can be implemented. In this section, we describe this formalization.

A scheme of proof rules for a DNN with $L$ layers is depicted in~\Cref{fig:CORArule}.
As abstract DNNs are generalizations of DNNs, we first reason about any DNN and its trivial abstraction. For that, we use the proof rule $\mathsf{triv-abs}$, based on~\eqref{eq:corabasecase}.
Recall that the correctness of the CORA abstraction is established based on~\Cref{prop:neuron-merging} and~\Cref{lem:abstraction-correctnes-conditions}, applied sequentially to all layers of the network. This yields the main part of the proof, depicted in the $\mathsf{base-abs}$ and the $\mathsf{l_k-abs}$ rules. $\mathsf{base-abs}$ established the correctness of CORA for the first layer based on any hyperrectangle $\bounds_1$ bounding the input property $\nnInputSet$. 
Then, using $\mathsf{l_k-abs}$ repeatedly on each layer of the abstract network. Then, we can conclude the correctness of CORA using the $\mathsf{CORA-L}$ rule for a DNN with $L$ layers. Then, these rules can be integrated into the proof construction scheme described in~\Cref{sec:fullproof}.
To ease notation, we assume that all networks' internal variables are well defined as in~\eqref{eq:dnn} and~\eqref{eq:rednn}. To ease notation further, we define the following:
\begin{align*}  
	\nn(\x)& = \phi_L(\nnW_L \cdots \phi_1(\nnW_1 \x + \nnb_1) \cdots  + \nnb_L), \\
	\redNN_0(\x)& = \nnActFun_L(\nnW_L \cdots\ \nnActFun_1(\nnW_1 \x + \nnb_1 \oplus \{\mathbf{0}\}) \cdots\  + \nnb_L \oplus \{\mathbf{0}\}),\\
	\redNN_k(\x)& = \nnActFun_L(\nnW_L \cdots \nnActFun_k( \rednnW_K \cdots \nnActFun_1(\rednnW_1 \x + \rednnb_1 \oplus \redBounds_1) \cdots \rednnb_k \oplus \redBounds_k) \cdots  + \nnb_L \oplus \{\mathbf{0}\}),\\
	\redNN(\x)& = \redNN_L(\x) = \nnActFun_L(\rednnW_L  \cdots\ \nnActFun_1(\rednnW_1 \x + \rednnb_1 \oplus \redBounds_1) \cdots\ + \rednnb_L \oplus \redBounds_L)
\end{align*}
and use the left hand side as an abbreviation of the right hand side.

\noindent\emph{Proof checking.} In order to check a proof witness for CORA abstraction, the checker needs to receive the original DNN verification query $\dnnvquery{\nn}{\nnInputSet}{Q}$, as well as the candidate abstract network $\redNN$. Then, any intermediate abstract network $\redNN_k$ can be constructed during the checking process based on $\redNN$ and $\nn$. Note that in contrary to CORA, the proof checker is not required to construct abstract networks. In addition, it could be implemented with formal guarantees of correctness (e.g., by using arbitrarily-precise arithmetic for its computation), prioritizing reliability over scalability.

\begin{figure}[h!]
	\centering
	\scalebox{0.95}{
		\rulename{triv-abs:}
		\drule{
			\nn
			\ \
			\redNN_0
			\ \ 
			\x\in\R^{\numNeurons_0}
		}
		{
			\nn(\x)\in\redNN_0(\x)
		}
		\ \ \ \ 
		\rulename{base-abs:}
		\drule{
			\redNN_0
			\ \ 
			\redNN_1
			\ \ 
			\x \in \nnInputSet \subseteq \redBounds_1
		}
		{
			\forall x \in\nnInputSet. \ \redNN_0(\x) \subseteq \redNN_1(\x)
		}
	} 
	
	\vspace{0.4cm}
	\scalebox{0.9}{
		\rulename{l_k-abs:}
		\drule{
			k\in\{2,\cdots,L-1\}
			\ \
			\bucket\subset[\numNeurons_k]
			\ \
			\redNN_k
			\ \ 
			\redNN_{k+1}
			\\ 
			\rednnW_k = \nnW_{k(\bar{B},\cdot)}
			\ \
			\rednnb_k = \nnb_{k(\bar{B})}
			\ \
			\rednnW_{k+1} = \nnW_{k+1(\cdot,\bar{B})}
			\ \
			\rednnb_{k+1} = \nnb_{k+1}
			\\
			\redBounds_{k} = \redBounds_{k(\bar{\bucket})}
			\ \ 
			\redBounds_{k+1} = \nnW_{k+1(\cdot,\bucket)}\bounds_{k(\bucket)}
		}
		{
			\forall x \in\nnInputSet. \ \redNN_k(\x) \subseteq \redNN_{k+1}(\x)
		}
	} 
	\vspace{0.4cm}
	
	\scalebox{0.9}{
		\rulename{CORA-L:}
		\drule{
			\nn(\x)\in\redNN_0(\x)
			\ \
			\forall x \in\nnInputSet. \ \redNN_0(\x) \subseteq \redNN_1(\x)
			\ \ 
			\cdots
			\ \ 
			\forall x \in\nnInputSet. \ \redNN_{L-1}(\x) \subseteq \redNN(\x)  
			\ \
			\mathcal{Q}\subset\R^{\numNeurons_\numLayers}
		}
		{
			\certpred(\dnnvquery{\redNN}{\nnInputSet}{Q}) \implies \certpred(\dnnvquery{\nn}{\nnInputSet}{Q})             
		}
	} 
	\caption{A scheme of proof rules for CORA abstraction of a DNN with $L$ layers}
	\label{fig:CORArule}
\end{figure}
\noindent \textbf{Example:} an example for a proof, proving the abstraction presented  in~\Cref{fig:abstraction-example}, appears in~\Cref{fig:proofexample}. The proof uses the concrete DNN $\examplenn$ and the abstract networks $\examplerednn_0,\examplerednn_1,\examplerednn$.  For simplicity, we omit the direct definitions of all matrices. However, we can see that the underlying weight matrices of $\examplerednn$ can be obtained by removing rows from the weight matrices of $\examplenn$. Furthermore, $\redBounds_2 = [-0.55,-0.45]$ is indeed obtained (in~\Cref{fig:abstraction-example}) by multiplying the input bounds and their weights, for the indices corresponding to the bucket, i.e., we indeed have that $\redBounds_2 = \nnW_{2(\cdot,\bucket)}\redBounds_{1(\bucket)} $.

\begin{figure}[h!]
	\centering
	\scalebox{0.9}{
		\drule{
			\begin{tabular}[b]{cccc}
				\drule{
					\examplenn \quad \examplerednn_0 \quad \x \in \mathbb{R}^{4}
				}{
					\examplenn(\x) \in \examplerednn_0(\x)
				}
				&
				\drule{
					\examplerednn_0 \quad \examplerednn_1 \quad \x \in \nnInputSet_1 \coloneq \redBounds_1
				}{
					\forall x \in \nnInputSet.\ \examplerednn_0(\x) \subseteq \examplerednn_1(\x)
				}
				&
				\scalebox{0.95}{
					\drule{
						\bucket=\{1,2,3\}\quad \examplerednn_1 \quad \examplerednn \\
						\rednnW_1 = \nnW_{1(\bar{\bucket},\cdot)}\quad \rednnb_1 = \nnb_{1(\bar{\bucket})} \\
						\rednnW_2 = \nnW_{2(\cdot,\bar{\bucket})}\quad \rednnb_2 = \nnb_{2} \\
						\redBounds_1 = \redBounds_{1(\bar{\bucket})}\quad \redBounds_2 = \nnW_{2(\cdot,\bucket)}\bounds_{1(\bucket)}
					}{
						\forall x \in \nnInputSet.\ \examplerednn_1(\x) \subseteq \examplerednn(\x)
					}
				}
				&
				$\mathcal{Q}_1 \subset \mathbb{R}$
			\end{tabular}
		}{
			\certpred(\dnnvquery{\examplerednn}{\nnInputSet_\text{1}}{Q_\text{1}}) \implies \certpred(\dnnvquery{\examplenn}{\nnInputSet_\text{1}}{Q_\text{1}})
		}
	}
	\caption{An example of a proof of an abstraction, for the DNN $\examplenn$ and the abstract network $\examplerednn$.}
	\label{fig:proofexample}
\end{figure}

\section{Related Work}\label{sec:related-work}

This work builds on two main pillars in DNN verification: proof production and abstraction.

\emph{Proof production} is a well-established area in formal verification, particularly within SAT and SMT solvers~\cite{HeJhMaMc04, NiPrReZoBaTi19, BaReKrLaNiNoOzPrViViZoTiBa22} among many others, where the generation of proofs or certificates serves to improve trust in automated verification results. These proofs can be independently checked, enhancing the reliability of verification pipelines. Despite its importance in traditional verification, proof production remains largely unexplored in the context of DNN verification, and most existing tools do not provide formal proofs as part of their output.

\emph{Abstraction}\cite{CoCo77,ClGrLo94,ClGrJhLuVe00} is a classical technique in formal verification, widely used to tackle scalability and complexity challenges. In the domain of DNN verification, abstraction has been actively studied through two main lenses. The first is \emph{abstract interpretation}, which over-approximates neural network behavior using abstract domains~\cite{GeMiDrTsChVe18, SiGePuVe19}. The second is \emph{abstraction refinement}, where the verifier incrementally refines the abstraction based on counterexamples or property violations, improving precision over time~\cite{ElGoKa22, AsHaKrMo20, ElCoKa22, OsBaKa22, CoElBaKa23, LaAl23, LiXiShSoXuMi24}. These techniques have proven effective in improving both scalability and verification success rates.

However, the intersection of abstraction and proof production has received very limited attention, in both directions: how proofs can influence abstraction, and how abstraction can contribute to proof construction. An early example of the former, in the SAT domain, is the work of~\cite{HeJhMaMc04}, where proofs are used to guide and refine the abstraction process.

As for the latter, our work is, to the best of our knowledge, among the first to investigate how abstraction mechanisms can be directly integrated into the production of formal proofs in the context of DNN verification. In a recent work~\cite{SiSaMeSi25}, a Domain Specific Language (DSL), designed for defining and certifying the soundness of abstract interpretation DNN verifiers, is introduced and evaluated over several DNN verifiers. This work is focused on proving DNN verifiers that employ linear over-approximations of activation functions, while our method focuses on separate proofs for neuron-merging abstraction process, and for the verification process. An interesting future work would be to formalize our scheme using the DSL of~\cite{SiSaMeSi25}.

\section{Conclusion and Future Work}\label{sec:conclusion}

This work-in-progress aims to bridge abstraction and proof production in DNN verification. On one hand, incorporating abstraction enhances the efficiency and compactness of proof production in the formal verification of neural networks. On the other hand, proofs can be generated for verification tools that apply abstraction to improve performance of reliable verifiers. To achieve this, we allow the abstraction process itself to be certified. By introducing a modular proof rule that separates the verification proof from the abstraction proof, we establish a foundation for generating complete proofs while using abstraction to aid in their construction. This modular approach allows integration with existing proof-producing verifiers for the verification component, while enabling the development of novel proof mechanisms specific to abstraction.
We presented 
  a general algorithm for abstraction-refinement-based proof production in DNN verification and demonstrated how it can be instantiated using current tools for both proof generation and abstraction.

The next steps of our work include the implementation and evaluation of our method with respect to proof size, verification time, and proof-checking time; over real-world benchmarks. 
Looking forward, we identify two promising directions for future work. First, integrating residual reasoning~\cite{ElCoKa22} could improve the effectiveness of abstraction refinement procedures. Second, leveraging abstract proofs within CDCL-based frameworks~\cite{IsReWuBaKa25,LiYaZhHu24} offers a compelling avenue for bridging abstraction and clause learning-based proof systems.

\section*{Acknowledgements}
The research presented in this paper was partially funded by the project FAI under project number 286525601 funded by the German Research Foundation (Deutsche Forschungsgemeinschaft, DFG). 

This work was partially funded by the European Union (ERC, VeriDeL, 101112713). Views and opinions expressed are however those of the author(s) only and do not necessarily reflect those of the European Union or the European Research Council Executive Agency. Neither the European Union nor the granting authority can be held responsible for them.

This work was performed in part using high-performance computing equipment obtained under NSF Grant \#2117377.

\bibliographystyle{splncs04}
\bibliography{main}
 
 \newpage
 \appendix

\end{document}